\documentclass[format=sigconf]{acmart}
\setcopyright{rightsretained}
\acmDOI{}

\acmISBN{}

\acmConference[Conference'19]{Conference'19, July 2019, Washington, DC, USA}
\acmYear{2018}
\copyrightyear{}

\acmArticle{}
\acmPrice{}

\usepackage[utf8]{inputenc}
\usepackage[T1]{fontenc}
\usepackage{amsmath,amssymb,amsthm}
\usepackage{xspace}
\DeclareFontFamily{U}{mathb}{\hyphenchar\font45}
\DeclareFontShape{U}{mathb}{m}{n}{ <-6> mathb5 <6-7> mathb6 <7-8>
  mathb7 <8-9> mathb8 <9-10> mathb9 <10-12> mathb10 <12-> mathb12 }{}
\DeclareSymbolFont{mathb}{U}{mathb}{m}{n}
\DeclareMathSymbol{\prec}{\mathrel}{mathb}{"A0}
\DeclareMathSymbol{\succ}{\mathrel}{mathb}{"A1}
\DeclareMathSymbol{\preceq}{\mathrel}{mathb}{"A8}
\DeclareMathSymbol{\succeq}{\mathrel}{mathb}{"A9}
\DeclareMathSymbol{\precneq}{\mathrel}{mathb}{"AC}
\DeclareMathSymbol{\succneq}{\mathrel}{mathb}{"AD}

\usepackage{color}
\definecolor{gray}{gray}{0.4}

\usepackage{amscd}
\usepackage{natbib}
\usepackage{enumerate}
\usepackage{algorithmic,algorithm}
\usepackage[english]{babel}
\usepackage[shortlabels]{enumitem} 

\usepackage{fixme}
\fxsetup{author={T},envlayout={color},%
  layout={footnote,marginclue},theme=color}

\FXRegisterAuthor{mf}{amf}{M}
\usepackage[mathscr]{euscript}

\usepackage{array}

\usepackage{mathtools}
\mathtoolsset{showonlyrefs,showmanualtags}

\newtheorem{theorem}{Theorem}[section]
\newtheorem{lemma}[theorem]{Lemma}
\newtheorem{proposition}[theorem]{Proposition}
\newtheorem{corollary}[theorem]{Corollary}
\newtheorem{definition}[theorem]{Definition}

\newtheorem{remark}[theorem]{Remark}

\newcommand{\LM}{\mathrm{LM}}%for leading term
\newcommand{\LT}{\mathrm{LT}}%for leading term
\newcommand{\ZZ}{\mathbb{Z}}
\newcommand{\KK}{\mathbb{K}}
\newcommand{\NN}{\mathbb{N}}

\newcommand{\divides}{\mid}
\newcommand{\LC}{\mathrm{LC}}%for leading coefficient

\newcommand{\lcm}{\mathrm{lcm}}

\newcommand{\sth}{\mathord{\color{black!57}\bullet}\;\!}
\newcommand{\bfu}{\mathbf{u}}
\newcommand{\bfe}{\mathbf{e}}
\newcommand{\bfp}{\mathbf{p}}

\newcommand{\bfT}{\mathbf{T}}

\newcommand{\bfalpha}{\boldsymbol{\alpha}}
\newcommand{\bfbeta}{\boldsymbol{\beta}}

\newcommand{\bfq}{\mathbf{q}}

\newcommand{\bfeps}{\mathbf{\epsilon}}

\newcommand{\modpol}[1]{\overline{#1}}
\newcommand{\bbfp}{\modpol{\bfp}}

\newcommand{\bbfe}{\modpol{\bfe}}

\newcommand{\sig}{\mathfrak{s}}
\newcommand{\Sig}{S}

\newcommand{\sreduc}{$\sig$-reducible\xspace}

\newcommand{\SPol}{S\textup{-Pol}}
\newcommand{\GPol}{G\textup{-Pol}}

\newcommand{\resp}{resp.\xspace}

\newcommand{\ie}{\emph{i.e.}\xspace}

\newcommand{\Syz}{\mathrm{Syz}}
\newcommand{\TSyz}{\mathrm{TSyz}}

\newcommand{\Mon}{\mathrm{Mon}}
\newcommand{\Ter}{\mathrm{Ter}}

\newcommand{\RR}{R}
\renewcommand{\AA}{A}

\begin{document}

\title{Signature-based Möller's algorithm\\
  for strong Gröbner bases over PIDs}

\author{Maria Francis}
\affiliation{
  \institution{Indian Institute of Technology Hyderabad}
  \city{Hyderabad, India}  
}
\email{mariaf@iith.ac.in}

\author{Thibaut Verron}
\affiliation{%Johannes Kepler Universität;
  \institution{Institute for Algebra / Johannes Kepler University}
  \city{Linz, Austria}  
}
\email{thibaut.verron@jku.at}

\thanks{This work was started when the first author was supported by the Austrian FWF grant Y464.
  The second author is supported by the Austrian FWF grant F5004.}



\begin{abstract}
  Signature-based algorithms are the latest and most efficient approach as of today to compute Gröbner bases for polynomial systems over fields.
  Recently, possible extensions of these techniques to general rings have attracted the attention of several authors.

  In this paper, we present a signature-based version of Möller's classical variant of Buchberger's algorithm for computing strong Gröbner bases over Principal Ideal Domains (or PIDs).
  It ensures that the signatures do not decrease during the algorithm, which makes it possible to apply classical signature criteria for further optimization.
  In particular, with the F5 criterion, the signature version of Möller's algorithm computes a Gröbner basis without reductions to zero for a polynomial system given by a regular sequence.
  We also show how Buchberger's chain criterion can be implemented so as to be compatible with the signatures.

  We prove correctness and termination of the algorithm.
  Furthermore, we have written a toy implementation in Magma, allowing us to quantitatively compare the efficiency of the various criteria for eliminating $S$-pairs.
\end{abstract}

% \begin{CCSXML}
%   <ccs2012>
%   <concept>
%   <concept_id>10010147.10010148.10010149.10010150</concept_id>
%   <concept_desc>Computing methodologies~Algebraic algorithms</concept_desc>
%   <concept_significance>500</concept_significance>
%   </concept>
%   </ccs2012>
% \end{CCSXML}

% \ccsdesc[500]{Computing methodologies~Algebraic algorithms}

%\vspace{-1mm}
%\ccsdesc[500]{Computing methodologies~Algebraic algorithms}
%\printccsdesc

\vspace{-1.5mm}
\terms{Algorithms, Theory}


\maketitle



\section{Introduction}
\label{sec:Introduction}

% Importance of Gröbner bases and also for rings
\paragraph{Motivation and main results}
Ever since Gröbner bases were introduced by Buchberger in 1965~\cite{Buchberger:2006:thesistranslation}, they have become a valuable tool for solving polynomial systems in many different applications, for example in cryptography or in engineering.
For many applications, restricting Gröbner basis computations to polynomials over a field is enough.
However, some applications require the computation of Gröbner bases over rings.
For instance, Gröbner bases over $\ZZ$ can be used in lattice-based cryptography~\cite{FrancisDukkipati:2014:Hash}, or as a multi-purpose tool in integer linear algebra~\cite{Lichtblau:applications}.

% Signatures for fields
In the case of polynomials over a field, many algorithms have been developed to make Gröbner basis computations more and more efficient.
The latest generation of Gröbner basis algorithms for fields is the class of signature-based algorithms.
They introduce signatures, which are defined as the leading terms of a module representation of polynomials in terms of the generators of the ideal.
This notion makes it possible to eliminate redundant computations and reductions of $S$-polynomials, by enforcing the key invariant that \emph{signatures always increase during the algorithm}.
With this information, algorithms are able to use criteria such as the F5 criterion~\cite{Faugere:2002:F5}, which allows to compute a Gröbner basis for an ideal given by a regular sequence without any reduction to zero.

Several algorithms have been developed for Gröbner bases over rings.
In~\cite{Moller:1988:grobnerrings2}, Möller sketched an algorithm for computing so-called weak Gröbner bases over general commutative rings (described in detail in \cite[Sec.4.2]{Adams:1994:introtogrobnerbasis}) and presented a specialized version, computing strong Gröbner bases over Principal Ideal Domains (PIDs).
In this paper, to avoid ambiguity, we call the former algorithm \emph{Möller's weak GB algorithm} and the latter  \emph{Möller's strong GB algorithm} (or \emph{Möller's algorithm} when clear from the context).

In this paper, we show how to add signatures to Möller's strong GB algorithm.
We prove that our signature-variant of the algorithm is able to compute a strong Gröbner basis of any polynomial ideal over a PID, and that the crucial invariant holds: the algorithm never encounters a signature smaller than that of a previously computed polynomial.

Möller's algorithm maintains a weak Gröbner basis $G_{w}$ and a strong Gröbner basis $G_{s}$.
The basis $G_{w}$ is obtained by reducing $S$-polynomials by elements of the strong basis; the basis $G_{s}$ is obtained by computing (but not reducing) $G$-polynomials (called $T$-polynomials in~\cite{Moller:1988:grobnerrings2}) of elements of the weak basis.

The signature version of Möller's algorithm maintains a signature of each element in $G_{w}$.
As for elements of $G_{s}$, requiring the computation of $G$-polynomials to maintain a matching signature is too restrictive.
However, we prove that maintaining an upper bound on their signature is sufficient to ensure that the signature of $S$-polynomials in $G_{w}$ does not drop when reduced by elements of $G_{s}$, and that the algorithm as a whole is correct.

Additional criteria can be implemented to further eliminate redundant $S$-polynomials, such as Buchberger's criteria~\cite{Buchberger:criteria}.
In particular, we show that Buchberger's chain criterion can be implemented in a similar fashion as Gebauer-Möller's criteria, with an order compatible with the selection strategy by smallest signature.
The fact that signatures do not drop implies that the algorithm is also compatible with additional criteria such as the singular criterion, the syzygy criterion or the F5 criterion.
We prove that the algorithm is correct and terminates.

We have written a toy implementation of Möller's algorithm with signatures\footnote{Available online: \url{https://github.com/ThibautVerron/SignatureMoller}} in the computer algebra system Magma~\cite{Magma}, and we use it to give experimental data on the number of computed and eliminated pairs for some systems.
We also discuss some optimizations which can be applied when implementing the algorithm.

\paragraph{Related work}

Signature-based algorithms for fields have a long history.
Early work in this direction was described in~\cite{MollerMoraTraverso1992}, where the authors use computations in a polynomial module for a similar purpose, and Algo.~F5~\cite{Faugere:2002:F5} showed that module computations can be avoided by considering only signatures.
From there, significant work has gone into studying signature-based algorithms from a theoretical standpoint and extending them.
An excellent survey of this is given~in~\cite{eder:2017:survey}.

% Signatures on rings, Eder
Several algorithms have been developed for Gröbner bases over rings.
Möller's work~\cite{Moller:1988:grobnerrings2}, on an algorithm for weak GBs over general rings and an algorithm for strong GBs over PIDs, was already mentioned.
It also gives a survey of precursor works regarding Gröbner bases over rings.
Similar ideas, notably $G$-polynomials, are present in different variations of Buchberger's algorithm for PIDs\cite{Pan:Dbases} or Euclidean domains~\cite{Lichtblau,Kandri-Rody-Kapur}.

Extending signature techniques to rings has been the focus of recent research, starting in 2017 with Eder and Popescu~\cite{Eder:2017:EuclideanRings}.
In that work, the authors consider a signature-based version of Gröbner basis algorithms for Euclidean domains. 
The authors showed with a counter-example that implementing \emph{totally ordered }signatures for rings cannot ensure that the crucial invariant holds.
However, their algorithm can detect signature drops and fall back to existing algorithms without signatures.
It can nonetheless serve as an efficient preprocessing step.

% Signatures on rings, us for Möller + main difference with Eder
In~\cite{FV2018}, we described a way to add signatures to Möller's weak GB algorithm, and proved that the resulting algorithm is correct and terminates over PIDs.
In particular, there is no signature drop in the algorithm, and additional criteria such as the F5 criterion can be used to eliminate reductions to zero in the case of a regular sequence.
The main difference with the approach of~\cite{Eder:2017:EuclideanRings} is that signatures are only partially ordered, and the coefficient parts of signatures are never compared.
% However, Möller's algorithm suffers from a significant combinatorial cost when computing $S$-polynomials, and cannot be used in practice for systems of even moderate size.

% This work: Möller technique for Buchberger with Gebauer-Möller criteria
In the present paper, we incorporate the same signature techniques into Möller's strong GB algorithm~\cite{Moller:1988:grobnerrings2}.

% The computation of saturated sets in Möller's weak GB algorithm takes time exponential in the size of the Gröbner basis, and Möller's strong GB algorithm replaces it with the computation of $S$-polynomials and $G$-polynomials, only involving pairs which can be computed in quadratic time.
% It uses pairs instead of saturated sets, and thus replaces a step with cost exponential in the size of the currently computed Gröbner basis with a step with quadratic cost.
The main ingredients for the proofs of correctness of the algorithm with signatures and criteria are the relation between regular weak $S$-polynomials and weak signature-Gröbner bases from~\cite{FV2018}, and the characterization of Gröbner bases in terms of syzygies of the leading terms, given by the Lifting Theorem~\cite[Th.~1]{Moller:1988:grobnerrings2}, which we generalize to a signature setting.

\section{Preliminaries}
\label{sec:Notat-prel}

\subsection{Notations}

Let $R$ be a principal ideal domain (PID), which is assumed to have a unit element and be commutative.
We assume that the ring $R$ is \emph{effective} in the sense that:
\begin{enumerate}
  \item there are algorithms for all arithmetic operations ($+$, $\ast$, comparison to zero and to one) in $R$;
  \item there is an algorithm which, given $a$ and $b \in R$, computes their greatest common divisor $d$ and the Bézout coefficients $u$ and $v$ such that $au + bv = d$;
  \item there is an algorithm which, given $a$ and $b \in R$, tests whether $a$ divides $b$ and if so, computes the quotient $b/a$.
\end{enumerate}

\begin{remark}
  Effective Euclidean rings (in the sense that there are algorithms for (1) and an algorithm for Euclidean division), thanks to the extended Euclid algorithm, are effective PIDs.
\end{remark}

Let $A = R[x_{1}, \dots,x_{n}]$ be the polynomial ring in $n$ indeterminates $x_{1},\dots,x_{n}$ over $R$.
A monomial in $A$ is $x^{a} := x_{1}^{a_{1}} \dots x_{n}^{a_{n}}$ where $a = (a_{1},\dots,a_{n}) \in \NN^{n}$.
A term in $A$ is $kx^{a}$, where $k \in R \setminus \{0\}$.
The set of terms (resp. monomials) of $A$ is denoted by $\Ter(A)$ (resp. $\Mon(A)$).

We use the notation $\mathfrak{a}$ for ideals in the polynomial algebra $A$ and $I$ for ideals in the coefficient ring $\RR$.

The notion of monomial order can be directly extended from $\KK[x_1, \ldots, x_n]$ to $\AA$.
In the rest of the paper, we assume that $\AA$ is endowed with an implicit monomial order $\prec$, and we define as usual the leading monomial $\LM$, the leading term $\LT$ and the leading coefficient $\LC$ of a given polynomial.

Given a tuple of polynomials $(g_{1},\dots,g_{s})$ and $i \in \{1,\dots,s\}$, we will frequently denote, for brevity, $M(i) = \LM(g_{i})$, $C(i) = \LC(g_{i})$ and $T(i) = \LT(g_{i}) = C(i) M(i)$.
Given $i,j \in \{1,\dots,s\}$, we will frequently denote $M(i,j) = \lcm(M(i),M(j))$, $T(i,j) = \lcm(T(i),T(j))$ and $C(i,j) = \lcm(C(i),C(j))$.

\subsection{Signatures}
\label{sec:Signatures}

We consider the free $\AA$-module $\AA^{m}$ with basis $\bfe_1, \ldots, \bfe_m$.
A term (\resp{} monomial) in $\AA^{m}$ is $kx^{a}\bfe_{i}$ (\resp{} $x^{a}\bfe_{i}$) for some $k \in \RR \setminus \{0\}$, $x^{a} \in \Mon(\AA)$, $i \in \{1,\dots,m\}$.
The set of terms of $\AA^{m}$ is denoted by $\Ter(A^{m})$.
In this paper, terms in $\AA^{m}$ are ordered using the Position Over Term (POT) order, defined by
\begin{displaymath}
  k x^a\bfe_i \prec lx^b\bfe_j \iff i \lneq j \text{ or } ( i = j \text{ and } x^a \prec x^b).
\end{displaymath}
Given two terms $kx^{a}\bfe_{i}$ and $lx^{b}\bfe_{j}$ in $\AA^{m}$, we write $kx^{a}\bfe_{i} \simeq lx^{b}\bfe_{j}$ if they are incomparable, \ie{} if $a=b$ and $i=j$.

Given a set of polynomials $f_{1},\dots,f_{m} \in \AA$, we define an $\AA$-module homomorphism $\bar{\cdot} : \AA^{m} \to \AA$, by setting $\bbfe_{i} = f_{i}$ and extending linearly to $\AA^{m}$.

We recall the concept of signatures in $\AA^m$.
Let $\bfp = \sum_{i=1}^{m} p_{i} \bfe_{i}$ be a module element.
Under the POT ordering, the signature of $\bfp$ is $\LT(p_{i})\bfe_{i}$ where $i$ is such that $p_{i+1}=\dots=p_{m}=0$ and $p_{i}\neq 0$.
Signatures are of the form $kx^a\bfe_i$, where $k \in \RR, x^a \in \mathrm{Mon}(\AA)$ and $\bfe_i$ is a standard basis vector.

Note that we have two ways of comparing two similar signatures $\sig(\bfalpha) = kx^a\bfe_i$ and $\sig(\bfbeta) = lx^b\bfe_j$.
We write $\sig(\bfalpha)=\sig(\bfbeta)$ if $k=l$, $a=b$ and $i=j$, and we write $\sig(\bfalpha) \simeq \sig(\bfbeta)$ if $a=b$ and $i=j$, $k$ and $l$ being possibly different.
If $\RR$ is a field, one can assume that the coefficient is $1$, and so this distinction is not important.

Note also that when we order signatures, we only compare the corresponding module monomials, and disregard the coefficients.
This is a different approach from the one used in \cite{Eder:2017:EuclideanRings}, where both signatures and coefficients are ordered.

% Given a tuple $(\bfalpha_{1},\dots,\bfalpha_{s})$ of module elements in $\AA^{m}$ and $i,j \in \{1,\dots,s\}$, we shall frequently denote $S(i) = \sig(\bfalpha_{i})$ for brevity.

\section{Algorithm}
\label{sec:Algorithm}

\subsection{Definitions}
\label{sec:Definitions}

Möller's algorithm for computing strong Gröbner bases over PIDs uses the classical constructions of $S$-polynomials and reductions, together with $G$-polynomials.
For each polynomial $f$, we want to keep track of a signature $\sig(f)$, such that $\sig(f) = \sig(\bfp)$ for some $\bfp \in A^{m}$ with $\bbfp = f$.
For that reason, the algorithm will maintain lists of \emph{labelled} polynomials, where the label encodes the information available regarding the signature.

\begin{definition}
  Let $f_{1},\dots,f_{m} \in A$, $\mathfrak{a} = \langle f_{1},\dots,f_{m} \rangle$, and $(f,l) \in \mathfrak{a} \times \Ter(A^{m})$.
  We say that $(f,l)$ is:
  \begin{itemize}
    \item a \emph{$S$-labelled polynomial}, with \emph{signature} $l$ if $l = \sig(\bfp)$ for some $\bfp \in A^{m}$ with $\bbfp = f$;
    \item a \emph{$G$-labelled polynomial}, with \emph{$G$-signature} $l$ if $l \succeq \sig(\bfp)$ for some $\bfp \in A^{m}$ with $\bbfp = f$.
  \end{itemize}
  By abuse of notation, we say that $f$ is $S$-labelled (resp. $G$-labelled) and we denote $\sig(f) := l$ (resp. $\sigma(f) := l$).
\end{definition}
\begin{remark}
  $S$-labelled polynomials are naturally $G$-labelled.
\end{remark}
\begin{remark}
  The base polynomials $f_{i}$ are naturally $S$-labelled with signature $\bfe_{i}$.
\end{remark}

We go through the required constructions, with the signature-related restrictions allowing to maintain the labelling, starting with $S$-polynomials and reductions:

\begin{definition}
  Let $G = \{g_{1},\dots,g_{t}\} \subset A$ be a set of $S$-labelled polynomials.
  For all $i \in \{1,\dots,t\}$, let $M(i)$, $T(i)$ and $C(i)$ be respectively $\LM(g_{i})$, $\LT(g_{i})$ and $\LC(f_{i})$.
  Given $i,j \in \{1,\dots,t\}$, let $M(i,j)$, $T(i,j)$ and $C(i,j)$ be respectively $\lcm(M(i),M(j))$, $\lcm(T(i),T(j))$ and $\lcm(C(i),C(j))$.

  The $S$-polynomial of $g_{i}$ and $g_{j}$ is the polynomial
  \begin{equation}
    \label{eq:1}
    \SPol(g_{i},g_{j}) = \frac{T(i,j)}{T(i)} g_{i} - \frac{T(i,j)}{T(j)} g_{j}.
  \end{equation}
  The leading term of its polynomial evaluation is $\precneq M(i,j)$.
  %The monomial $M(i,j)$ is called the \emph{term degree} of the pair $(i,j)$.
  
  The \emph{$S$-pair} $(i,j)$ is called \emph{regular} if
  $\frac{M(i,j)}{M(i)} \sig(g_{i}) \neq \frac{M(i,j)}{M(j)} \sig(g_{j})$ and \emph{singular} otherwise.
  The $S$-pair $(i,j)$ is called \emph{strictly singular} if
  $\frac{T(i,j)}{T(i)} \sig(g_{i}) = \frac{T(i,j)}{T(j)} \sig(g_{j})$, and \emph{admissible} otherwise.
  Note that regular pairs are admissible.

  Let $(i,j)$ be an admissible $S$-pair, we extend the $S$-labelling of $G$ to $\SPol(g_{i},g_{j})$ by defining $\sig(\SPol(g_{i},g_{j}))=S(i,j)$, defined as:
  \begin{enumerate}
    \item $S(i,j) = \max\left( \frac{T(i,j)}{T(i)} \sig(g_{i}) , \frac{T(i,j)}{T(j)} \sig(g_{j}) \right)$
    if $(i,j)$ is a regular $S$-pair;
    \item $S(i,j) = \left( \frac{C(i,j)}{C(i)} - \frac{C(i,j)}{C(j)} \right) \frac{M(i,j)}{M(i)} \sig(g_{i})$
    if $(i,j)$ is a singular, non strictly singular, $S$-pair.
    % \item $\precneq \max\left( \frac{T(i,j)}{T(i)} \sig(f_{i}) , \frac{T(i,j)}{T(j)} \sig(f_{j}) \right)$
    % if $(i,j)$ is a strictly singular $S$-pair (not admissible).
  \end{enumerate}
\end{definition}
\begin{remark}
  If $(i,j)$ is not an admissible $S$-pair, it is strictly singular, and knowing the signature of $g_{i}$ and $g_{j}$ is not enough to know a signature for $\SPol(g_{i},g_{j})$.
  All we know is that $S(i,j) \succneq \sig(\bfp)$ for some $\bfp \in A^{m}$ with $\bbfp = \SPol(g_{i},g_{j})$.
  Such a situation is called a \emph{signature drop}.
\end{remark}
  
  % \begin{multline}
  %     S(i,j) := \sig(\SPol(f_{i},f_{j}))\\
  %   =
  %   \begin{cases}
  %     \max\left( \frac{T(i,j)}{T(i)} \sig(f_{i}) , \frac{T(i,j)}{T(j)} \sig(f_{j}) \right) \text{ if $(i,j)$ is a regular $S$-pair} \\
  %     \left( \frac{C(i,j)}{C(i)} - \frac{C(i,j)}{C(j)} \right) \frac{M(i,j)}{M(i)} \sig(f_{i}) \text{ if $(i,j)$ is a singular, non strictly singular, $S$-pair} \\
  %     \tau < \max\left( \frac{T(i,j)}{T(i)} \sig(f_{i}) , \frac{T(i,j)}{T(j)} \sig(f_{j}) \right) \text{ if $(i,j)$ is a strictly singular $S$-pair.}
  %   \end{cases}
  % \end{multline}
\begin{definition}
  Let $G = \{g_{1},\dots,g_{t}\} \subset A$ be a set of $G$-labelled polynomials, let $f \in A$ be a $S$-labelled polynomial and let $g \in A$.
  We say that $f$ \emph{(strongly) $\sig$-reduces in one step} to $f$ modulo $F$ if there exists $g_{i} \in F$ such that
  \begin{enumerate}
    \item $\LT(g_{i})$ divides $\LT(f)$, say $\LT(f) =  c \mu \LT(g_{i})$ with $c \in R$ and $\mu \in \Mon(A)$;
    \item $g = f - c \mu g_{i}$;
    \item $\mu \sigma(g_{i}) \preceq \sig(f)$
  \end{enumerate}
  We say that $f$ \emph{(strongly) regular reduces in one step} to $g$ modulo $F$ if the signature inequality is strict: $x^{a} \sigma(g_{i}) \precneq \sig(f)$.

  We say that $f$ $\sig$-reduces (\resp regular reduces) to $g$ modulo $G$ if $g$ is the result of a sequence of successive $\sig$-reductions (resp. regular reductions) in one step from $f$.

  If $g$ is the result of \emph{regular} reducing $f$ modulo $G$, then we can extend the $S$-labelling to $g$ by letting $\sig(g) = \sig(f)$.
\end{definition}

Using those definitions, we recall the definition of a (strong) signature Gröbner basis.
\begin{definition}
  \label{def:strong-GB}
  Let $f_{1},\dots,f_{m} \in A$, and $G = {g_{1},\dots,g_{t}}$ a set of $G$-labelled polynomials in $\langle f_{1},\dots,f_{m} \rangle$.
  Let $\bfT \in \Ter(A^{m})$, the set $G$ is called a \emph{(strong) $\sig$-Gröbner basis up to signature $\bfT$} if for all $g \in \langle f_{1},\dots,f_{m} \rangle$ with signature $\preceq \bfT$, $g$ (strongly) $\sig$-reduces to $0$ modulo $G$.
  \footnote{In the literature, it is sometimes only required that all elements with signature $\precneq \bfT$ $\sig$-reduce to $0$.}
  It is called a \emph{strong $\sig$-Gröbner basis} if it is a strong $\sig$-GB up to signature $\bfT$ for all $\bfT \in \Ter(A^{m})$.
\end{definition}

Next, we recall the definition of GCD-polynomials (or $G$-polynomials for short) \footnote{In the literature, $G$-polynomials are sometimes called $T$-polynomials~\cite{Moller:1988:grobnerrings2}.} and how to equip them with a $G$-labelling.

\begin{definition}
  Let $f \in A$ be a $G$-labelled polynomial, and $g \in A$ a $S$-labelled polynomial, such that $\LT(f) = a \mu$, $\LT(g) = b \nu$, with $a,b \in R$, $\mu, \nu \in \Mon(A)$.
  Let $d = \gcd(a,b)$ and $u$ and $v$ be the Bézout coefficients such that $u a + v b = d$.
  The $G$-polynomial of $f$ and $g$ is the module element
  \begin{equation}
    \label{eq:3}
    \GPol(f,g) = u \frac{\lcm(\mu,\nu)}{\mu} f + v \frac{\lcm(\mu,\nu)}{\nu} g.
  \end{equation}
  The leading term of its polynomial evaluation is $d\, \lcm(\mu,\nu)$.

  We extend the $G$-labelling by defining the $G$-signature of $\GPol(f,g)$ to be
  \begin{equation}
    \label{eq:27}
    \sigma(\GPol(f,g)) := S_{G}(f,g) = \max\left( \frac{\lcm(\mu,\nu)}{\mu} \sigma(f), \frac{\lcm(\mu,\nu)}{\nu} \sig(g) \right).
  \end{equation}
\end{definition}
Since we do not require that the pair be admissible in any sense, this is really only a $G$-labelling.
However, we will prove that this $G$-labelling for $G$-polynomials preserves enough information regarding the signature of the polynomials participating in the construction (Lem.~\ref{lemme:completion-decomposition}), and that it is sufficient to ensure that subsequent reductions preserve the signature, which is a key point in proving that the algorithm is correct.

\subsection{Algorithm}
\label{sec:Algorithm-1}

\begin{algorithm}
  \caption{Möller's algorithm with signatures}
  \label{algo:sigBuchberger}
  \begin{algorithmic}
    \STATE \textbf{Input}
    $\{f_{1},\dots,f_{m}\} \subset A=R[x_{1},\dots,x_{n}]$, $R$ a PID
    \STATE \textbf{Output}
    % \begin{itemize}
    %   \item $G_{w}% =\{g_{1},\dots,g_{N}\}
    %   $ a weak $\sig$-Gröbner basis of $\langle f_{1},\dots,f_{m} \rangle$ 
    %   \item
    $G_{s}$ a set of $G$-labelled polynomials in $A$, which is a (strong) $\sig$-Gröbner basis of $\langle f_{1},\dots,f_{m} \rangle$
    \STATE \textbf{Local variables}
    \begin{itemize}[leftmargin=10pt]
      \item $G_{w} = \{g_{1},\dots,g_{r}\}$ a set of $S$-labelled polynomials in $A$, which is a weak Gröbner basis of
      $\langle f_{1},\dots,f_{m} \rangle$
      \item $\mathcal{P} \subset \NN^{2}$ a set of admissible $S$-pairs
    \end{itemize}
    %\end{itemize}
    \vspace{5pt}
    \STATE $G_{s}, G_{w}, \mathcal{P} \leftarrow \emptyset$
    % \COMMENT{At all times, $G_{w} = \{g_{1},\dots,g_{r}\}$}
    \FOR{$i \in \{1,\dots,m\}$}
    \STATE $\textsf{Update}(G_{w},G_{s},\mathcal{P},f_{i},\bfe_{i})$
    \WHILE{$\mathcal{P} \neq \emptyset$}
    \STATE Pick and remove $(i,j)$ from $\mathcal{P}$ with minimal $S(i,j)$
    \STATE $g \leftarrow \mathsf{SPol}(g_{i},g_{j})$
    \STATE $\textsf{Update}(G_{w},G_{s},\mathcal{P},g,S(i,j))$
    \ENDWHILE
    \ENDFOR
    \STATE Return $G_{s}$
  \end{algorithmic}
\end{algorithm}

\begin{algorithm}
  \caption{Procedure \textsf{Update}: update the weak and the strong Gröbner bases, and the list of pairs, eliminating pairs with Buchberger's chain criterion and signature restrictions}
  \label{algo:2}
  \begin{algorithmic}
    \STATE \textbf{Input} $G_{w} \subset A$ set of $S$-labelled polynomials, $G_{s} \subset A$ set of $G$-labelled polynomials, $\mathcal{P} \subset \NN^{2}$, $f \in A$, $\sig(f) \in \Ter(A^{m})$
    \\[5pt]
    \STATE $g \leftarrow \mathsf{RegularReduce}(f,\sig(f),G_{s})$
    \IF{$g \neq 0$}
    \STATE $r \leftarrow \#G_{w}{+}1;\; g_{r} \leftarrow g$ \COMMENT{Index of the new element}
    \STATE $G_{w} \leftarrow G_{w} \cup \{(g_{r},\sig(f))\}$
    \STATE $G_{s} \leftarrow G_{s} \cup \{(g_{r},\sig(f))\}$
    \FORALL{$h \in G_{s}$}
    \STATE $G_{s} \leftarrow G_{s} \cup (\GPol(h,g_{r}), \Sig_{G}(h,g_{r}))$
    \ENDFOR

    % \STATE $G_{s} \leftarrow G_{s} \cup \{(g,\sigma(f))\} \cup \{\mathsf{GPol}(g,h) : h \in G_{w}\}$
    % \STATE $\mathsf{newpairs} \leftarrow \{(i,r) : i \in \{1..r-1\}\}$
    % \FOR {$(i,j) \in \mathsf{newpairs}$ such that $\critM(i,j)$ or $\critF(i,j)$ holds}
    % \STATE Remove $(i,j)$ from $\mathsf{newpairs}$
    % \ENDFOR
    % \FOR {$(i,r) \in \mathsf{newpairs}$ such that $\critM(i,j)$ or $\critF(i,j)$ holds}
    % \STATE Remove $(i,j)$ from $\mathsf{newpairs}$
    % \ENDFOR
    \FORALL{$i \in \{1,\dots,r{-}1\}$ \\
      \leavevmode\phantom{\textbf{for all}}\textbf{such that} $(i,r)$ is an admissible $S$-pair\\
      \leavevmode\phantom{\textbf{for all}}\textbf{and} $\forall\, k \in \{1,\dots,r{-}1\}$, $\mathsf{Chain}(i,r;k)$ does not hold,}
    \STATE Add $(i,r)$ to $\mathcal{P}$
    \ENDFOR
    \FORALL{$(i,j) \in \mathcal{P}$ \textbf{such that} $\mathsf{Chain}(i,j;r)$ holds}
    \STATE Remove $(i,j)$ from $\mathcal{P}$ 
    \ENDFOR
    \ENDIF
  \end{algorithmic}
\end{algorithm}

Möller's algorithm with signatures is presented in Algo.~\ref{algo:sigBuchberger}.
It is a straightforward adaptation of Möller's algorithm, extended to keep track of the signature of computed polynomials, similar to the generic algorithm described in \cite{eder:2011:signature}.
Note that any time the algorithm mentions a $S$-labelled polynomial $f$ (resp. a $G$-labelled polynomial $f$), it means a pair $(f,\sig(f))$ (resp. a pair $(f,\sigma(f))$).

Algo.~\ref{algo:sigBuchberger} maintains two sets of generators, $G_{w}$ which will be a weak $\sig$-Gröbner basis and $G_{s}$ which will be a (strong) $\sig$-Gröbner basis.
The basis $G_{s}$ is the \emph{completion} of $G_{w}$, defined as follows.
\begin{definition}
  Let $F \subset A$ be a non-empty finite set of $G$-labelled polynomials, the completion $C(F)$ of $F$ is the set of $G$-labelled polynomials defined recursively as:
  \begin{itemize}
    \item $C(f) = \{f\}$;
    \item $C(f_{1},\dots,f_{r}) = \left\{ \GPol(g,f_{r}) : g \in C(f_{1},\dots,f_{r-1})\right\}.$
  \end{itemize}
\end{definition}
It is known that over a PID, the completion of a weak Gröbner basis is a strong Gröbner basis~\cite[Cor. after Th.~4]{Moller:1988:grobnerrings2}, we will prove in Cor.~\ref{cor:1} that it also holds for $\sig$-Gröbner bases.

Most of the book-keeping work, maintaining the bases and the list of pairs to consider together with signature information, is delegated to the subroutine \textsf{Update} (Algo.~\ref{algo:2}).
The most important feature of this subroutine is that it implements the following restrictions, which ensure that we can maintain a $S$-labelling in $G_{w}$:
\begin{enumerate}
  \item all reductions have to be \emph{regular} (that is, the signatures of reducers have to be strictly less than the signature of the reducee);
  \item all $S$-pairs have to be \emph{admissible} (that is, the signatures must not be an exact match);
  \item no restriction on $G$-pairs.
\end{enumerate}
We shall prove in Sec.~\ref{sec:Proofs-corr-term} that with those restrictions, the algorithm is correct and terminates.

The routine \textsf{RegularReduce} implements regular strong reduction modulo the already computed basis, due to space constraints it is not presented in details.

Additionally, Buchberger introduced two criteria to make the algorithm more efficient by eliminating $S$-polynomials: the coprime criterion~\cite[Sec.~2.10, Prop.~1]{Cox15} and the chain criterion~\cite[Sec.~2.10, Prop.~8]{Cox15}%
\footnote{In older editions of that book, those criteria can be found in Sec.~2.9, Prop.~4 and Prop.~10 respectively.}%
.
Implementing the coprime criterion is straightforward and not detailed here.
In order to implement the chain criterion, we use ideas similar to Gebauer and Möller's implementation~\cite{gebauer:1988:installation}, adapted to our selection order by smallest signatures first.
\begin{definition}
  Let $\{g_{1},\dots,g_{t}\} \subset A$ be a set of $S$-labelled polynomials.
  Let $(i,j,k) \in \{1,\dots,t\}^{3}$, we say that $\mathsf{Chain}(i,j;k)$ holds if
  \begin{equation}
    \label{eq:18}
    T(k) \divides T(i,j) \text{ and } S(i,j) \succeq \frac{T(i,j)}{T(k)} \sig(g_{k}).
  \end{equation}
\end{definition}
The consequence of that criterion is that $S$-pairs $(i,j)$ such that $\mathsf{Chain}(i,j,r)$ holds for some $r$ can be removed from consideration.

The criterion is also implemented as part of the \textsf{Update} subroutine (Algo.~\ref{algo:2}).

Similar to what was done with the signature-version of Möller's weak GB algorithm~\cite{FV2018}, further criteria can be added to the algorithm to make the computations more efficient: polynomials which have been regular reduced by are $1$-singular reducible can be eliminated, and the Syzygy, the F5 and the Singular criteria can eliminate redundant polynomials before any reduction.
In particular, the F5 criterion ensures that the algorithm does not perform any reduction to 0 for polynomial systems given as a regular sequence.
Due to space constraints, we refer to~\cite{FV2018} for details.

\section{Tools for the proofs}
\label{sec:Tools-proofs}

The rest of the paper will be devoted to proving that Algo.~\ref{algo:sigBuchberger} is correct and terminates.
In this section, we recall necessary definitions for the proofs in Sec.~\ref{sec:Proofs-corr-term}.

\subsection{Weak Gröbner bases}
\label{sec:Weak-Grobner-bases}

The main ingredient of the proof will be the fact that Möller's algorithm with signatures ensures that $G_{w}$ is a weak Gröbner basis.
In this section, we briefly recall relevant definitions and results.

\begin{definition}
  \label{def:weak-red}
  Let $f,g_{1},\dots,g_{s},h \in A$.
  We say that $f$ \emph{weakly (top) reduces} in one step to $h$ modulo $g_{1},\dots,g_{s}$ if there exists $J \subset \{1,\dots,s\}$ such that
  \begin{itemize}
    \item for all $i \in J$, there exists $x^{a_{i}} \in \Mon(A)$ such that $x^{a_{i}}\LM(g_{i}) = \LM(f)$
    \item there exists $c_{i} \in A, i \in J$ such that $\sum_{i \in J} c_{i}\LC(g_{i}) = \LC(f)$
    \item $h = f - \sum_{i \in J} c_{i}x^{a_{i}} g_{i}$.
  \end{itemize}
  In particular, $\LT(h) \precneq \LT(f)$.

  If $f$ is $S$-labelled and $g_{1},\dots,g_{s}$ are $G$-labelled, we call the one-step reduction a
  \begin{itemize}
    \item \emph{weak $\sig$-reduction} if for all $i \in J$, $x^{a_{i}}\sigma(g_{i}) \preceq \sig(f)$, and a
    \item \emph{regular weak $\sig$-reduction} if for all $i \in J$, $x^{a_{i}}\sigma(g_{i}) \precneq \sig(f)$.
  \end{itemize}
  As in the case of strong reductions, the terminology extends to sequences of reductions in one step.
\end{definition}

% \begin{definition}
%   \label{def:weak-GB}
%   Let $f_{1},\dots,f_{m} \in A$, $G \subset A^{m}$.
%   Let $\bfT \in \Ter(A^{m})$, the set $G$ is called a \emph{weak $\sig$-Gröbner basis up to signature $\bfT$} if for all $g \in \langle f_{1},\dots,f_{m} \rangle$ with signature $\precneq \bfT$, $g$ weakly $\sig$-reduces to $0$ modulo $G$.
%   It is called a \emph{weak $\sig$-Gröbner basis} if it is a weak $\sig$-GB up to signature $\bfT$ for all $\bfT \in \Ter(A^{m})$.
% \end{definition}

Weak Gröbner bases (resp. weak $\sig$-Gröbner bases) are defined as strong Gröbner bases (resp. strong $\sig$-Gröbner bases), replacing strong reductions (resp. strong $\sig$-reductions) with weak ones.

Weak Gröbner bases can be computed with Möller's weak GB algorithm~\cite[Algo.~4.2.1]{Adams:1994:introtogrobnerbasis}.
A signature version of this algorithm, for PIDs, was presented in~\cite{FV2018}.
This algorithm is similar to Buchberger's algorithm, but it replaces strong reductions with weak reductions and strong $S$-polynomials with weak $S$-polynomials, defined as follows in the context of PIDs.

\begin{definition}
  \label{def:weak-Spol}
  Let $g_{1},\dots,g_{t} \in A$ be $S$-labelled polynomials.
  Let $J$ be a subset of $\{1,\dots,t\}$, define $M(J) = \lcm(\{M(j) : j \in J\})$.
  Let $s \in J$ and $J^{\ast} = J \setminus \{s\}$.
  We say that $J$ is \emph{regular saturated}, with \emph{signature index} $s$, if
  \begin{equation}
    J^{\ast} = \left\{j \in \{1,\dots,t\} : M(j) \divides M(J)
    \text{ and } 
    \frac{M(J)}{M(s)}\sig(g_{s})\right\}.
    \label{eq:28}
  \end{equation}

  % Let $s \in J$ and $J^{\ast} = J \setminus \{s\}$.
  Let $c \in R$ be such that
  % \begin{equation}
  %   \label{eq:20}
    $\langle c \rangle = \langle C(j) : j \in J^{\ast} \rangle : \langle C(s) \rangle.$
  %\end{equation}
  Then there exists $(b_{j})_{j \in J^{\ast}}$ such that $c C(s) = \sum_{j \in J^{\ast}} b_{j}C(j)$ and the \emph{regular weak $S$-polynomial} associated to $J$ and $(b_{j})$ is
  \begin{equation}
    \label{eq:21}
    c \frac{M(J)}{M(s)} g_{s} - \sum_{j \in J^{\ast}} b_{i} \frac{M(J)}{M(s)}.
  \end{equation}
  This weak $S$-polynomial can be $S$-labelled with signature $S(J) = c \frac{M(J)}{M(s)}\sig(g_{s})$.
  %
  % If $g_{1},\dots,g_{s}$ are labelled polynomials, the $S$-polynomial is called \emph{regular} if for all $j \in J^{\ast}$, $\frac{M(J)}{M(j)} \sig(g_{j}) \precneq \frac{M(J)}{M(s)} \sig(g_{s})$.
  % The index $s$ is then called the \emph{signature index} of the saturated set $J$.\fxnote{This doesn't really depend on signatures. Maybe we should always give a name to $s$, even without signatures? Something like ``pivot index''?}
\end{definition}

\begin{comment}
We need:
- saturated sets (?)
- weak reductions
- weak S-polynomials
- weak S-GB
- some theorems from the old draft:
* 3.19 and 3.20 (GM or MM): weak GB with weak reductions + G-pols = strong GB with strong reductions
* 3.21 (GM or MM): weak S-pol as combination of strong S-pol
* 6.3 and 6.4: weak S-GB + G-pol = strong S-GB
* 6.5: 3.21 with signatures
- some theorems from the arxiv paper:
* 5.5 (correctness, maybe just cite without making explicit)
* Probably 5.6 (termination)
\end{comment}

\subsection{Syzygies}
\label{sec:Syzygies}

\begin{comment}
- definition of syzygies, homogeneous syzygies, principal syzygies, term degree of a syzygy
- regular syzygies, strictly singular syzygies? => NO
- S-polynomial associated to a syzygy
- characterization of a GB in terms of syzygies
\end{comment}

A crucial tool for the proofs will be the syzygy characterization of Gröbner bases, using the syzygy lifting theorem of Möller~\cite{Moller:1988:grobnerrings2}.
This characterization gives a framework for proving that criteria eliminating $S$-pairs do not break the correctness or termination of the algorithm.
The central notion is that of term-syzygies, of which we recall the definition.\footnote{In the literature, term-syzygies are sometimes simply called syzygies, and syzygy polynomials, $S$-polynomials.}

\begin{definition}
  Let $G = (g_{1},\dots,g_{t})$ be a tuple of nonzero $S$-labelled polynomials in $A$.
  We consider the free module $A^{t}$ with basis $\bfeps_{1},\dots,\bfeps_{t}$.
  For any element $\Sigma = \sum_{i=1}^{t} s_{i}\bfeps_{i} \in A^{t}$, we define
  % \begin{equation}
  %   \label{eq:88}
  \(\modpol{\Sigma} = \sum_{i=1}^{t} s_{i}g_{i}\).
  % \end{equation}
  We say that $\Sigma$ is a \emph{term-syzygy} of $G$ if
  \begin{equation}
    \label{eq:93}
    \LT(\modpol{\Sigma}) \precneq \max\{ \LT(s_{i})T(i) : i \in \{1,\dots,t\}\}.
  \end{equation}
  The polynomial $\modpol{\Sigma}$ is called the \emph{syzygy polynomial} of $\Sigma$.
  
  The set of all term-syzygies of $G$ is denoted by $\TSyz(G)$, it is a submodule of $A^{t}$ called the \emph{syzygy module} of $\LT(G)$.
  
  If there exists a monomial $\mu$ s.t. for all $i \in \{1,\dots,t\}$, $\LM(s_{i}g_{i}) = \mu$ or $0$, the term-syzygy $\Sigma$ is called \emph{homogeneous with term degree $\mu$}.

  The \emph{signature} of $\Sigma$ is $\sig(\Sigma) = \max_{i}\{s_{i}\sig(g_{i})\}$.
  
  A tuple $(\Sigma_{1},\dots,\Sigma_{s})$ of $\TSyz(G)$ is called a \emph{$\Sig$-basis} of $\TSyz(G)$ if for all $\Sigma \in \TSyz(G)$, there exists $p_{1},\dots,p_{s} \in A$ such that
  \begin{itemize}
    \item $\Sigma = \sum_{i=1}^{s} p_{i}\Sigma_{i}$ 
    \item $\sig(\Sigma) \succeq \max_{i}\{\LM(p_{i})\sig(\Sigma_{i})\}$.
  \end{itemize}
\end{definition}

% \begin{remark}
  
%   % We adopt more specific names to avoid future conflicts with reductions to zero, corresponding to syzygies of the whole polynomials as opposed to just their leading term, and strong and weak $S$-polynomials.
% \end{remark}it's

\begin{definition}
  A strong (resp. weak) $S$-polynomial is the syzygy polynomial $\modpol{\Sigma}$ for some homogeneous term-syzygy $\Sigma \in \Syz(F)$.
  We call those syzygies strong (resp. weak) $S$-pol. syzygies.
  
  Strong $S$-pol. syzygies are homogeneous term-syzygies with exactly two non-zero coefficients, and are sometimes called \emph{principal} term-syzygies in the literature.
\end{definition}

The characterization of Gröbner bases using term-syzygies is given in Möller's lifting theorem~\cite[Th.~4]{Moller:1988:grobnerrings2}, of which we give a signature version here.

\begin{theorem}
  \label{thm:lifting}
  Let $\mathfrak{a} = \langle f_{1},\dots,f_{m} \rangle$ be an ideal in $A$ and $G=(g_1,\ldots, g_t)$ be a tuple of nonzero $S$-labelled polynomials in $\mathfrak{a}$ such that for all $i \in \{1,\dots,m\}$, $f_{i}$ $\sig$-reduces to $0$ modulo $G$.
  Let $\bfT \in \Ter(A^{m})$, and let $\TSyz_{\bfT}(G)$ be the module of term-syzygies generated by term-syzygies with signature at most $\bfT$.

  Let $\Sigma_{1},\dots,\Sigma_{s} \in \TSyz(G)$ be a homogeneous $\Sig$-basis of $\TSyz_{\bfT}(G)$,
  where $\Sigma_{i} = \sum_{j=1}^{t} \sigma_{ij} \bfeps_{j}$, and define for $i \in \{1,\dots,s\}$ the syzygy polynomial
  % \begin{equation}
  %   \label{eq:64}
  \(\modpol{\Sigma_{i}} = \sum_{j=1}^{t} \sigma_{ij} g_{j}\).
  % \end{equation}

  Then % $G$ is a weak $\sig$-Gröbner basis of $\mathfrak{a}$ up to signature $\bfT$ if and only if for all $i \in \{1,\dots,s\}$, $\modpol{\Sigma_{i}}$ weakly $\sig$-reduces to $0$ modulo $G$;
  % \item 
  $G$ is a strong $\sig$-Gröbner basis of $\mathfrak{a}$ up to signature $\bfT$
  if and only if for all $i \in \{1,\dots,s\}$, $\modpol{\Sigma_{i}}$ strongly $\sig$-reduces to $0$ modulo $G$.
\end{theorem}
\begin{proof}
  The proof is similar to that of~\cite[Th.~1 and Th.~4]{Moller:1988:grobnerrings2}: indeed, if $f \in \mathfrak{a}$ has signature $\bfT \in \Ter(A^{m})$, $f$ has a representation $\sum_{i=1}^{m} q_{i}f_{i}$ with $\max_{i} \LT(q_{i})\bfe_{i} \preceq \bfT$.
  Since all $f_{i}$'s $\sig$-reduce to $0$ modulo $G$, $f$ also has a representation $\sum_{j=1}^{t} h_{j}g_{j}$ such that $\max_{i} \LT(h_{i})\sig(g_{i}) \preceq \bfT$.

  Following the proof of~\cite[Th.~1]{Moller:1988:grobnerrings2} allows to use term-syzygies with signature $\preceq \bfT$ to rewrite this representation into a Gröbner representation, that can be decomposed into a sequence of reductions.

  Conversely, if all $f \in \mathfrak{a}$ $\sig$-reduce to $0$, in particular it is true for the syzygy polynomials of term-syzygies of $G$.
\end{proof}

% The idea of Gröbner basis algorithms is, given a set of polynomials $G$, to compute syzygy polynomials corresponding to a homogeneous basis of term-syzygies of $G$, reduce them modulo $G$, add any non-zero result to $G$, and repeat until there is nothing left to add.

% In the classical Buchberger algorithm, the basis of term-syzygies is given by all principal term-syzygies.
% Buchberger and Gebauer-Möller's criteria state that some of those term-syzygies can be safely removed while maintaining the fact that the basis generates all term-syzygies, which proves that they are correct.

% Similarly, Möller's algorithm for weak Gröbner bases uses all homogeneous term-syzygies as a basis of term-syzygies, and Gebauer-Möller's criteria again state that some of those term-syzygies can be safely excluded.

\section{Correctness and termination}
\label{sec:Proofs-corr-term}

% In this section, we prove that Algo.~\ref{algo:sigBuchberger} is correct and terminates.

% The proofs will proceed in 4 steps:
% \begin{enumerate}
%   \item regular homogeneous term-syzygies (corresponding to regular weak $S$-polynomials) form a homogeneous basis of term-syzygies;
%   \item any regular homogeneous term-syzygy can be written as a linear combination of (not necessarily regular) $S$-pairs without an increase in the term degree and signature, and in particular $S$-pairs form a basis of term-syzygies;
%   \item strong $S$-polynomials all weakly $\sig$-reduce to zero modulo $G_{w}$;
%   \item strong $S$-polynomials all strongly $\sig$-reduce to zero modulo $G_{s}$.
% \end{enumerate}

% This proves that Buchberger's algorithm with signatures, without any criterion, is correct.
% Termination is proved separately.\ifxfatal{Maybe by noting that Buchberger's algorithm without sigs is known to terminate?}

% Proving the correctness of Gebauer-Möller's criteria requires proving, in addition to (2), that:
% \begin{enumerate}
%   \item[(2.5)] any non strictly-singular $S$-pair can be written as a linear combination of non strictly-singular $S$-pairs satisfying neither $\critM$, $\critF$ nor $\critB$.
% \end{enumerate}

%\ifxfatal{TODO: Introduction to the proofs}

\subsection{Signature properties}
\label{sec:Signature-properties}

In this subsection, we prove useful lemmas, related to the behavior of signatures throughout the algorithm, and generalizing with signatures the correspondence between weak and strong constructions (reductions and $S$-polynomials) described in~\cite{Moller:1988:grobnerrings2}.

\begin{comment}
- Signatures are non-decreasing in G_w
- Signatures are non-decreasing in G_s
- Weak S-polynomials to strong S-polynomials with signatures

Definition of completion
C(F U {f}) = C(F) U {f} U {GPol(g,f) : g in C(F)}
\end{comment}

\begin{lemma}
  \label{lemme:nondecreasing-buchberger-weak}
  Let $\{g_{1},\dots,g_{r}\}$ be the value of $G_{w}$ at any point in the course of Algo.~\ref{algo:sigBuchberger}.
  Then $\sig(g_{1}) \preceq \sig(g_{2}) \preceq \dots \preceq \sig(g_{r})$.
\end{lemma}
\begin{proof}
  The proof is similar to that of~\cite[Lem.~5.2]{FV2018}.
  Assume that there exists $i$ such that $\sig(g_{i}) \succ \sig(g_{i+1})$ and that $i$ is the smallest index with this property.
  Let $(j_{i},k_{i})$ (resp. $(j_{i+1},k_{i+1})$) be the admissible pair used to compute $g_{i}$ (resp. $g_{i+1}$).

  If $i$ is not one of $j_{i+1},k_{i+1}$, then $(j_{i+1},k_{i+1})$ was already in the queue $\mathcal{P}$ when $(j_{i},k_{i})$ was selected, and so, by the selection criterion in the algorithm, $\Sig(j_{i},k_{i}) \prec \Sig(j_{i+1},k_{i+1})$.

  If $i$ is either $j_{i+1}$ or $k_{i+1}$, wlog we can assume that $i = j_{i+1}$.
  Then
  \begin{align}
    \label{eq:31}
    \Sig(j_{i+1},k_{i+1})
    &\simeq \max \left( \frac{T(i,k_{i+1})}{\LT(g_{i})}\sig(g_{i}),
     \frac{T(i,k_{i+1})}{\LT(g_{k_{i+1}})}\sig(g_{k_{i+1}}) \right) \\
    &\succeq \frac{T(i,k_{i+1})}{\LT(g_{i})}\sig(g_{i}) \succeq \sig(g_{i}).\qedhere
  \end{align}
\end{proof}

It allows us to prove that the signatures of elements in $G_{s}$ are also non-decreasing.
\begin{lemma}
  \label{lemme:nondecreasing-buchberger-strong}
  Let $\{g_{1},\dots,g_{r-1}\}$ be the value of $G_{w}$ at any point in the course of Algo.~\ref{algo:sigBuchberger}, and let $g_{r}$ be the next computed element in the basis.
  Then all elements added to $G_{s}$ have $G$-signature $\succeq \sig(g_{r})$.

  More generally, all elements added to $G_{s}$ in later steps have $G$-signature $\succeq \sig(g_{r})$.
\end{lemma}
\begin{proof}
  The elements added to $G_{s}$ in the call to \textsf{Update}
  with $g_{r}$ as new element, are $g_{r}$ (with signature $\sig(g_{r})$) and
  all $G$-polynomials $\GPol(h,g_{r})$ for $h$ already in $G_{s}$ (with $G$-signature $\Sig_{G}(\sigma(h),\sig(g_{r}))$).
  Those $G$-labelled polynomials all have $G$-signature $\succeq \sig(g_{r})$.

  The generalized statement follows from the fact that $\sig(g_{s}) \succeq \sig(g_{r})$ for $s > r$ (Lem.~\ref{lemme:nondecreasing-buchberger-weak}).
\end{proof}

The next lemma is a more precise description of elements of $G_{s}$.
\begin{lemma}
  \label{lemme:completion-decomposition}
  Let $G_{w} = \{g_{1},\dots,g_{r}\}$ be a set of $S$-labelled polynomials, and $G_{s}$ be its ($G$-labelled) completion.
  Let $h \in G_{s}$, then there exists $i_{1},\dots,i_{k} \in \{1,\dots,r\}$ such that
  \begin{equation}
    \label{eq:6}
    h = \GPol(\GPol(\cdots\GPol(g_{i_{1}},g_{i_{2}}),\dots,g_{i_{k-1}}),g_{i_{k}}).
  \end{equation}
  Furthermore, there exists $c_{j} \in R$, $m_{j} \in \Mon(A)$, $j \in \{1,\dots,k\}$
  such that $\LT(h) =\sum_{j=1}^{k}c_{j}m_{j}T(i_{j})$
  and $\sigma(h) \simeq \max(m_{j}\sig(g_{i_{j}}))$.
\end{lemma}
\begin{proof}
  The existence of $i_{1},\dots,i_{k}$ and the decomposition of $h$ and $\LT(h)$ are clear by definition of the completion.

  For the inequality regarding the signature, we proceed by induction on $k$, where the base case $k=1$ is clear.

  Let $k > 1$, and let $h_{k-1}$ be the result of the innermost $k-1$ $G$-polynomials in the expansion of $h$.
  So $h = \GPol(h_{k-1},g_{i_{k}})$ and $h_{k-1}$ expands as $k-1$ successive $G$-polynomials of $g_{i_{1}},\dots,g_{i_{k-1}}$, with $m'_{j} M(i_{j}) = \LM(h_{k-1})$ for all $j \in \{1,\dots,k-1\}$.
  Note that for all $j \in \{1,\dots,k-1\}$, $\mu m'_{j} = m_{j}$.
  
  There exists $\mu \in \Mon(A)$ such that $\LM(h) = \mu \LM(h_{k-1}) = m_{k}M(i_{k})$, and
  \begin{align}
    \label{eq:7}
    \sigma(h) &\simeq \max(\mu \sigma(h_{k-1}),m_{k}\sig(g_{i_{k}})) \text{ by def. of the $G$-signature} \\
    \label{eq:8}
    &\simeq \max\left( \mu \max_{j \leq k-1}(m'_{j}\sig(g_{i_{j}})), m_{k}\sig(g_{i_{k}}) \right) \text{ by induction hyp.} \\
    &\simeq \max_{j \leq k}(m_{j} \sig(g_{i_{j}})).\qedhere
  \end{align}
\end{proof}

The last results of this section generalize the correspondence between weak and strong Gröbner bases~\cite{Moller:1988:grobnerrings2}, adding some control over the signatures.
First, we generalize the equivalence between weak reduction and strong reduction through completion of the reducers~\cite[Prop.~2]{Moller:1988:grobnerrings2}.
\begin{lemma}
  \label{lemme:weak-reduc-iff-strong-reduc}
  Let $G_{w} = \{g_{1},\dots,g_{r}\}$ be a weak $\sig$-GB up to signature $\bfT$, and $G_{s}$ be its completion.
  Let $f$ be a $S$-labelled polynomial with signature $\sig(f) \prec \bfT$, then the following properties are equivalent:
  \begin{enumerate}
    \item $f$ is weakly \sreduc (resp. weakly regular \sreduc) mod. $G_{w}$;
    \item $f$ is strongly \sreduc (resp. strongly regular \sreduc) mod. $G_{s}$.
  \end{enumerate}
\end{lemma}
\begin{proof}
  For $(1) \Rightarrow (2)$, we proceed by induction on $r$.
  The case $r=1$ is clear, because then both $G_{w}$ and $G_{s}$ contain only the element $g_{1}$.

  For the general case, let $f$ be a $S$-labelled polynomial with signature $\sig(f) \prec \bfT$ and weakly \sreduc modulo $G_{w}$.
  Let $H_{w} = \{g_{j} : j \in J \subseteq \{1,\dots,r\}\} \subseteq G_{w}$ be a set of weak $\sig$-reducers of $f$, and consider its completion $H_{s} = C(H_{w}) \subseteq G_{s}$.
  By \cite[Prop.~2]{Moller:1988:grobnerrings2}, $f$ is strongly reducible modulo $H_{s}$.
  Let $h \in H_{s}$ be a strong reducer of $f$.
  In particular, there exists $\mu \in \Mon(A)$ such that $\mu\LM(h) = \LM(f)$.
  In order to prove that $h$ is a strong $\sig$-reducer of $f$, we need to prove that $\mu\sigma(h) \preceq \sig(f)$.

  By Lem.~\ref{lemme:completion-decomposition}, $h$ expands as iterated $G$-polynomials of elements $h_{1},\dots,h_{k}$ of $H_{w}$ such that for all $j \in \{1,\dots,k\}$, there exists $m_{j} \in \Mon(A)$ such that $m_{j}\LM(h_{j}) = \LM(h)$ and $\sigma(h) = \max(m_{j}\sig(h_{j}))$.

  Let $j \in \{1,\dots,k\}$.
  Since $h_{j}\in H_{w}$, it is a weak $\sig$-reducer of $f$, so there exists $\mu_{j}$ such that $\mu_{j}\LM(h_{j}) = \LM(f)$, and $\mu_{j}\sig(h_{j}) \preceq \sig(f)$.
  Note that $\mu_{j} = m_{j}\mu$.
  So
  \begin{align}
    \label{eq:9}
    \mu \sigma(h) \simeq \mu \max(m_{j}\sig(h_{j})) \simeq \max(\mu_{j} \sig(h_{j})) \preceq \sig(f).
  \end{align}

  The fact that $(2) \Rightarrow (1)$ is an immediate consequence of Lem.~\ref{lemme:completion-decomposition}: if $h \in G_{s}$ is a strong $\sig$-reducer of $f$, then it expands as iterated $G$-polynomials of elements $g_{i_{1}},\dots,g_{i_{k}} \in G_{w}$ which are weak $\sig$-reducers of $f$.
  
  The statements with regular $\sig$-reductions are proved similarly, replacing $\preceq$ with $\precneq$ throughout.
\end{proof}

As a consequence, like in~\cite{Moller:1988:grobnerrings2}, the completion of a weak $\sig$-GB is a strong $\sig$-GB.
\begin{corollary}
  \label{cor:1}
  Let $G_{w}=\{g_{1},\dots,g_{r}\}$ be a set of $S$-labelled polynomials, and $G_{s}$ its ($G$-labelled) completion.
  Let $\bfT \in \Ter(A^{m})$.
  Then
  \begin{itemize}
    \item $G_{w}$ is a weak $\sig$-GB up to signature $\bfT$ iff $G_{s}$ is a strong $\sig$-GB up to signature $\bfT$;
    \item $G_{w}$ is a weak $\sig$-GB iff $G_{s}$ is a strong $\sig$-GB.
  \end{itemize}
\end{corollary}

The last lemmas of this section generalizes the expression of a weak $S$-polynomial in terms of strong $S$-polynomials, with control over the signatures.
First, we take care of weak $S$-polynomials, without any regularity assumption.
\begin{lemma}
  \label{lemme:weak-SPol-to-strong}
  Let $(g_{1},\dots,g_{r})$ be a tuple of $S$-labelled polynomials. % and consider the $A$-module $A^{r}$ with basis $(\epsilon_{i})_{i \in \{1,\dots,r\}}$.
  Let $J \subset \{1,\dots,r\}$, and let $\bfp \subset A^{r}$ (with basis $(\epsilon_{j})$) be a homogeneous term syzygy associated to a weak $S$-pol. with support $J$.
  Then there exists coefficients $a_{i,j} \in R$, and monomials $m_{i,j}$, $i {<} j \in J$, such that
  \begin{equation}
    \label{eq:11}
    \bfp = \sum_{i,j \in J} a_{i,j}m_{i,j} \SPol(\bfeps_{i},\bfeps_{j}).
  \end{equation}
  In this decomposition:
  \begin{enumerate}
    \item for all $i,j \in J$, $m_{i,j}M(i,j) = M(J)$
    \item for all $i,j \in J$, $m_{i,j}S(i,j) \preceq \max(\frac{M(J)}{M(i)}\sig(g_{i}))$.
  \end{enumerate}
\end{lemma}
\begin{proof}
  The existence of $a_{i,j}$ and $m_{i,j}$, $i{<}j \in J$, is given by~\cite[Th.~2 and Prop.~1]{Moller:1988:grobnerrings2},
  and it follows from that proof that $m_{i,j}M(i,j) = M(J)$.
  So for all $i,j \in J$,
  \begin{equation}
  m_{i,j}S(i,j) = \frac{M(J)}{M(i,j)}S(i,j) \preceq \frac{M(J)}{M(i,j)} \frac{M(i,j)}{M(i)} \sig(f_{i}) \simeq \frac{M(J)}{M(i)}\sig(f_{i}),\label{eq:13}
\end{equation}

  and similarly for $j$.
\end{proof}

\begin{lemma}
  \label{lemme:weak-reg-SPol-to-strong}
  Let $(g_{1},\dots,g_{r})$ be a tuple of $S$-labelled polynomials.
  Let $J \subset \{1,\dots,r\}$ be a regular subset, with signature index $s$, and let $J^{\ast} = J \setminus \{s\}$.
  Let $\bfp \subset A^{r}$ be a homogeneous term syzygy associated to a regular weak $S$-polynomial.
  With the notations of~\ref{lemme:weak-SPol-to-strong}, denote $a_{i} := a_{i,s}$ if $i{<}s$ and $a_{s,i}$ otherwise, and define similarly $m_{i,j}$, so that we have the decomposition
  \begin{equation}
    \label{eq:2}
    \bfp = \sum_{i \in J^{\ast}} a_{i}m_{i} \SPol(\bfeps_{i},\bfeps_{s})
    + \sum_{i,j \in J^{\ast}} a_{i,j}m_{i,j} \SPol(\bfeps_{i},\bfeps_{j})
  \end{equation}
  In this decomposition:
  \begin{enumerate}
    % \item $\forall i \in J^{\ast}$, $m_{i}M(i,s) = M(J)$
    \item $\sum_{i \in J^{\ast}} a_{i}C(i,s) = C(J)$
    \item $\forall i \in J^{\ast}$, the $S$-pair $(i,s)$ is regular and $m_{i}S(i,s) \simeq S(J)$
    \item $\sum_{i \in J^{\ast}} a_{i}m_{i}S(i,s) = S(J) = \sig(\bfp)$
    \item $\forall i,j \in J^{\ast}$, $m_{i,j}\,S(i,j) \precneq S(J)$
  \end{enumerate}
\end{lemma}
\begin{proof}
  In the proof of~\cite[Prop.~1]{Moller:1988:grobnerrings2} $a_{i}$ and $m_{i}$, for $i \in J^{\ast}$, are defined as follows.
  Let $c$ be the generator of $\langle C(i) : i \in J^{\ast} \rangle : \langle C(s) \rangle$, and for $i \in J^{\ast}$, let $d_{i} = \frac{C(i,s)}{C(s)}$.
  Then there exists $(a_{i})_{i \in J^{\ast}}$, such that $c = \sum_{i \in J^{\ast}} a_{i}d_{i}$.
  In particular, $C(J) = \sum_{i \in J^{\ast}} a_{i}C(i,s)$.
  For $i \in J^{\ast}$, define $m_{i} = \frac{M(J)}{M(i,s)}$.
  With those $a_{i}$ and $m_{i}$, property 1 is satisfied.

  Since the set $J$ is regular with signature index $s$, by definition, $S(J) \simeq \frac{M(J)}{M(s)} \sig(f_{s})$, and for all $i \in J^{\ast}$, $\frac{M(J)}{M(s)}\sig(f_{s}) \succneq \frac{M(J)}{M(s)} \sig(f_{i})$.
  So for all $i \in J^{\ast}$, $\frac{M(i,s)}{M(s)}\sig(f_{s}) \succneq \frac{M(i,s)}{M(s)} \sig(f_{i})$, so the $S$-pair $(i,s)$ is regular and $S(J) \simeq \frac{M(J)}{M(i,s)} S(i,s) = m_{i}S(i,s)$.
  This proves property 2.

  By definition, $S(J) = \frac{C(J)}{C(s)}{M(J)}{M(s)}\sig(f_{s})$ and for all $i \in J^{\ast}$, $S(i,s) = \frac{C(i,s)}{C(s)}{M(J)}{M(s)}\sig(f_{s})$.
  So, expanding $C(J) = \sum_{i \in J^{\ast}} a_{i}C(i,s)$ again, property 3 is satisfied.

  Now consider $\bfq = \sum_{i,j \in J^{\ast}} a_{i,j}m_{i,j} \SPol(\bfeps_{i},\bfeps_{j})$.
  It corresponds to a homogeneous term syzygy, with term degree $\simeq M(J)$.
  We have seen above that for all $i \in J^{\ast}$, $\frac{M(J)}{M(s)}\sig(f_{s}) \succneq \frac{M(J)}{M(s)} \sig(f_{i})$.
  From Lem.~\ref{lemme:weak-SPol-to-strong}, for all $i,j \in J^{\ast}$, $m_{i,j} S(i,j) \preceq \max(\frac{M(J)}{M(s)} \sig(f_{i})) \precneq S(J)$.
\end{proof}

\begin{remark}
  Property (4) actually gives another proof of property (3), by proving that $\bfq$ has signature $\precneq \sig(\bfp)$.
  Writing $\bfq = \bfp - \sum_{i \in J^{\ast}} a_{i}m_{i} \SPol(\bfeps_{i},\bfeps_{s})$, it means that the signature of the two terms of the difference have to cancel out.
\end{remark}

\subsection{Proof of the algorithm}
\label{sec:Proof-correctness}

\begin{comment}
- Remind (and cite) the theorem about correctness 
- Theorem about correctness
\end{comment}

The proof of correctness makes use of the following result for weak signature Gröbner bases, proved in~\cite{FV2018}.
\begin{proposition}[{\cite[Th.~5.5]{FV2018}}]
  \label{prop:FV2018-correctness}
  Let $G_{w}=\{g_{1},\dots,g_{r}\}$ be a set of $S$-labelled polynomials.
  Let $\bfT \in \Ter(A^{m})$.
  Assume that all regular weak $S$-polynomials with signature $\preceq \bfT$ $\sig$-reduce to $0$ modulo $G_{w}$.
  Then $G_{w}$ is a weak signature Gröbner basis up to signature $\bfT$.
\end{proposition}

\begin{corollary}
  \label{cor:S-basis}
  Let $G=\{g_{1},\dots,g_{t}\}$ be a set of $S$-labelled polynomials, $\bfT \in \Ter(A^{m})$,
  \begin{equation}
    \label{eq:23}
    \mathcal{S}_{\prec \bfT}(G) = \left\{ \text{homo. term-syz. of $G$ with sig. $\precneq \bfT$} \right\}
  \end{equation}
  and
  \begin{equation}
    \label{eq:19}
    \mathcal{S}_{\bfT}(G) = \mathcal{S}_{\prec \bfT}(G)
    \cup \left\{ \text{regular weak $S$-pol. syz. of $G$ with sig. $\simeq \bfT$} \right\}.
  \end{equation}
  Then $\mathcal{S}_{\bfT}(G)$ is a $S$-basis of $\TSyz_{\bfT}(G)$.
\end{corollary}
\newcommand{\Aext}{A_\mathrm{ext}}
\newcommand{\Gext}{G_\mathrm{ext}}
\begin{proof}
  The notion of $S$-basis of term-syzygies only depends on the leading terms and labels of the family $G$.
  Extend the polynomial algebra $A=R[x_{1},\dots,x_{n}]$ into $\Aext=R[x_{1},\dots,x_{n},y_{1},\dots,y_{t}]$, with a block order ordering the $x_{i}$'s first according to the monomial order on $A$.
  Consider the set $\Gext = \{g_{i} - y_{i}\} \subset \Aext$, where $g_{i}-y_{i}$ is given the signature $\sig(g_{i})$.
  $S$-bases of syzygies of $\TSyz_{\bfT}(\Gext)$ and $\TSyz_{\bfT}(G)$ are in natural one-to-one correspondence.

  Let $\Sigma \in \TSyz_{\bfT}$.
  If $S(\Sigma) \precneq \bfT$ there is nothing to prove, so assume that $S(\Sigma) \simeq \bfT$.
  Write $\Sigma = \sum_{i=1}^{t} \sigma_{i} \epsilon_{i}$, $\bar{\Sigma} = \sum_{i=1}^{t} \sigma_{i}g_{i}$ and $\Sigma(y) = \sum_{i=1}^{t} \sigma_{i}y_{i}$, in particular the syzygy polynomial associated to $\Sigma$ in $\Aext$ is $\bar{\Sigma} - \Sigma(y)$.

  Let $S_{1},\dots,S_{k}$ be the regular weak $S$-pol. syzygies of $\Gext$ with signature $\simeq \bfT$.
  Regular reducing them, in $\Aext^{t}$, yields module elements of the form $S'_{i} = S_{i} - \sum \text{(elements with sig. $\precneq \bfT$)}$.
  Note that since we are only performing regular reductions and the signature of $S_{i}$ is not divisible by any $y_{j}$, those module elements remain linear in $y$.
  By Prop.~\ref{prop:FV2018-correctness}, adding to $G$ all the $S'_{i}$ ensures that all polynomials with signature at most $\bfT$ $\sig$-reduce to $0$, in particular, the syzygy polynomial of $\Sigma$ (in $\Aext$) $\sig$-reduces to $0$.
  In other words, there exist $\tau_{1},\dots,\tau_{k} \in \Ter(A)$ such that
  \begin{equation}
    \label{eq:4}
    \bar{\Sigma} - \Sigma(y) = \sum_{i=1}^{k} \tau_{i} \left( \bar{S'_{i}} - S'_{i}(y) \right) \text{ in $\Aext$}
  \end{equation}
  and, again since the reduction cannot increase the signature, the equality also holds in $A$:
  $\bar{\Sigma} = \sum_{i=1}^{k} \tau_{i} \bar{S_{i}}$ in $A$.
  So in the end, we get that
  \begin{equation}
    \label{eq:12}
    \Sigma(y) = \sum_{i=1}^{k} S'_{i}(y) = \sum_{i=1}^{k} S_{i}(y) + \sum \text{(elements with sig. $\precneq \bfT$)},
  \end{equation}
  and substituting back $y_{i} \leftarrow \epsilon_{i}$ gives a representation of $\Sigma$ as a linear combination of elements of $\mathcal{S}_{\bfT}$, where all summands have signature at most $\bfT = S(\Sigma)$. 
\end{proof}

\begin{theorem}[Correctness and termination of Algo.~\ref{algo:sigBuchberger}]
  \label{thm:correctness}
  Given $f_{1},\dots,f_{m} \in A$, Algo.~\ref{algo:sigBuchberger} terminates and returns a strong $\sig$-Gröbner basis of $\mathfrak{a} = \langle f_{1},\dots,f_{m} \rangle$.
\end{theorem}
\begin{proof}
  The proof of termination is a transposition of that of~\cite[Th.~5.6]{FV2018} (which follows the proof of termination in \cite{practicalgrobner:2012:stillman}), to prove that $G_{w}$, and thus $G_{s}$, cannot grow infinitely large.
  
  As for correctness, let $G_{w}$ and $G_{s}$ be as computed by Algo.~\ref{algo:sigBuchberger}.
  Assume that $G_{s}$ is not a strong $\sig$-GB of $\mathfrak{a}$, then there exists $\bfu \in \Ter(A^{m})$ such that $G_{s}$ is not a $\sig$-GB up to signature $\bfu$.
  Assume that $\bfu$ is minimal for this property, in particular, for all $\bfT \precneq \bfu$, $G_{s}$ is a strong $\sig$-GB up to signature $\bfT$.

  Equivalently, from Cor.~\ref{cor:1}, $G_{w}$ is a weak $\sig$-GB up to signature $\bfT$ but not a weak $\sig$-GB up to signature $\bfu$.
  By Cor.~\ref{cor:S-basis}, $\mathcal{S}_{\bfu}(G_{w})$ is a $\Sig$-basis of the module $\TSyz_{\bfu}(G_{w})$.
  Let $\mathcal{S}_{\prec} = \mathcal{S}_{\prec \bfu}(G_{w})$.
  Then by Lem.~\ref{lemme:weak-reg-SPol-to-strong}, the set
  \begin{equation}
    \label{eq:24}
    \mathcal{S}_{\prec}
    \cup \left\{ \text{regular strong $S$-pol. sygygies of $G_{w}$ with sig. $\simeq \bfu$} \right\}
  \end{equation}
  is a $\Sig$-basis of the module $\TSyz_{\bfu}(G_{w})$.

  Let $\Sigma(i,j)$ be a strong $S$-pol. syzygy associated with an $S$-pair $(i,j)$ such that Criterion $\mathsf{Chain}(i,j;k)$ holds for some $k \in \NN$.
  Then as in the classical case~\cite[Sec.~2.10, Prop.~8]{Cox15}, $\Sigma(i,j)$ can be rewritten as
  \begin{equation}
    \label{eq:25}
    \Sigma(i,j) = \frac{T(i,j)}{T(i,k)} \Sigma(i,k) - \frac{T(i,j)}{T(j,k)} \Sigma(j,k).
  \end{equation}
  The signature condition in $\mathsf{Chain}$ implies that this rewriting does not make the signature increase.
  So $\Sigma(i,j)$ can be removed from the $\Sig$-basis of term-syzygies.

  Iterating the process, we get that the set
  \begin{equation}
    \label{eq:26}
    \mathcal{S}_{\prec}
    \cup \left\{ \text{regular $S$-pairs of $G_{w}$ with sig. $\simeq \bfu$ not excluded by $\mathsf{Chain}$} \right\}
  \end{equation}
  is a $\Sig$-basis of the module $\TSyz_{\bfu}(G_{w})$.

  The algorithm ensures that all regular strong $S$-polynomials obtained from a $S$-pair not excluded by $\mathsf{Chain}$ strongly $\sig$-reduce to $0$ modulo $G_{s}$.
  Furthermore, by minimality of $\bfu$, for all syzygies $\Sigma$ in $\mathcal{S}_{\prec}$, the syzygy-polynomial $\bar{\Sigma}$ strongly $\sig$-reduces to zero modulo $G_{s}$.
  So all syzygy-polynomials associated with all term-syzygies in our basis strongly $\sig$-reduce to $0$ modulo $G_{s}$, and by the lifting theorem~\ref{thm:lifting}, $G_{s}$ is a strong $\sig$-Gröbner basis up to signature $\bfu$.
\end{proof}

\addtolength{\dbltextfloatsep}{-0.7cm}
\addtolength{\dblfloatsep}{-0.5cm}

\begin{table*}
  \centering
  \begin{tabular}{crrrrrrrr}
    \hline
    System & Pairs & $S$-pols &  Coprime & Chain & F5 & Singular & 1-Singular & Red. to 0 \\
    \hline
    Katsura-3 & 504 & 178 & 157 & 153 & 115 & 1 & 6 & 0\\
    Katsura-4 & 1660 & 603 & 509 & 517 & 388 & 9 & 84 & 0 \\
    Generic (3;2;10) & 383 & 192 & 73 & 99 & 117 & 1 & 19 & 0 \\
    Generic (3;3;5) & 2211 & 1161 & 155 & 911 & 842 & 0 & 78 & 0\\
    \hline
  \end{tabular}
  \caption{Experimental data on Möller's algorithm with signatures.}
  \label{tab:experimental}
\end{table*}

\section{Implementation and future work}
\label{sec:Exper-results-persp}

We have written a toy implementation\footnote{Available online: \url{https://github.com/ThibautVerron/SignatureMoller}} in Magma~\cite{Magma} of the algorithm, with the F5, Singular and $1$-singular criteria.
We give experimental data related to the computation of Gröbner bases for various polynomial systems over $\ZZ$: Katsura-$n$ systems, and random systems with fixed degree and size of the coefficients.
The data is given in Table~\ref{tab:experimental} (``Generic $(n;d;s)$'' is a random system of $n$ polynomials in $n$ variables with degree $d$ and coefficients in $[-s;s]$).
For each system, we give the number of considered $S$-pairs and reduced $S$-polynomials, as well as how many polynomials were excluded by the Coprime or Chain criterion (before being considered as a $S$-pair), by the F5 or Singular criterion (counted in $S$-pairs, not in $S$-polynomials), or because they are $1$-singular reducible (after regular reducing).
We also give the number of reductions to 0 appearing in the algorithm, which is 0 as expected for regular sequences.

Möller's weak GB algorithm involved a combinatorial bottleneck with cost exponential in the size of the current basis, making it impractical as soon as the basis exceeds 30 elements.
Möller's strong GB algorithm for PIDs replaces it with the computations of $S$-pairs, with quadratic cost.
As a result, the algorithm is faster, but nonetheless becomes slow as the basis grows.
As is frequently the case with Gröbner basis algorithms, the main bottleneck appears to be the reduction step.

We implemented two additional optimizations, for $\ZZ$, in order to reduce the size of the basis.
The first one is a heuristic at the selection step in the algorithm: when we pick a pair $(i,j)$ with minimal signature $S(i,j)$, we typically have a choice between many such pairs.
Selecting the one with the smallest coefficient part (in absolute value) appears to help eliminating subsequent $S$-polynomials faster, and makes the algorithm significantly faster: for instance, the Katsura-4 example was impractical before this change, and terminates in less than 30s after.

The second optimization relies on the following idea: for a given $i \in \{1,\dots,m\}$, when we enter the ``for'' loop at index $i$, we know that all subsequent polynomials will have a signature of the form $\sth \bfe_{k}$ with $k \geq i$, and all preceding polynomials have a signature of the form $\sth \bfe_{k}$ with $k < i$.
In particular, we do not need to consider the individual signatures of already computed elements, beyond the information that this signature is $\precneq \bfe_{i}$.

As such, we may inter-reduce the strong basis $G_{s}$ and replace both $G_{w}$ and $G_{s}$ with the result, all elements being given signature $\bfe_{1}$.
For this inter-reduction step, at least in the case of $\ZZ$, we could use Magma's highly optimized routines.

The consequence is that after each pass through the ``for'' loop, the weak and strong bases are made shorter, which slows down the growth of the list of pairs in the remainder of the algorithm.

One difficulty arising when computing signature Gröbner bases over rings is that the Singular criterion requires the signature to match exactly, including their coefficient.
This leads to the computation of many polynomials having similar signatures and leading monomials.
The heuristic presented above helps mitigate the issue, but it will be the object of future work to examine whether the Singular criterion can be extended to eliminate more elements, in the case of principal rings.

For computations over $\ZZ$ or $K[X]$, it would also be interesting to use the additional structure of an euclidean ring to make the computations faster.
It will be the focus of future research to investigate whether leading coefficient reductions~\cite{Kandri-Rody-Kapur,Lichtblau} can be added to the algorithm without breaking signature invariants.

\par\medskip\noindent\textbf{Acknowledgements}
\label{sec:Acknowledgements}
The authors thank C.~Eder for helpful suggestions, M.~Ceria and T.~Mora for a fruitful discussion on the syzygy paradigm for Gröbner basis algorithms, and M.~Kauers for his valuable insights and comments all through the elaboration of this work.

\bibliographystyle{plain}
\bibliography{bib2}

% \clearpage

% \section{Appendix: note to the reviewers}

% We would like to mention that the article~\cite{FV2018} is currently under review at another venue.
% Despite the very similar titles, the two articles present different results, about two different algorithms described by Möller in~\cite{Moller:1988:grobnerrings2}.

% In~\cite{FV2018}, we added signatures to Möller's weak GB algorithm, restricted to the case of PIDs.
% The algorithm uses (regular) weak $S$-polynomials and (regular) weak reductions, and does not implement Buchberger's criteria.
% Furthermore, it only computes weak Gröbner bases, even though they can be made strong by taking the completion.

% The present article is a follow-up on that work.
% Here, we add signatures to Möller's strong GB algorithm for PIDs.
% This algorithm uses the fact that the coefficient ring is a PID to replace (regular) weak $S$-polynomials with (admissible) strong $S$-polynomials, and (regular) weak reductions with (regular) strong reductions by a completion of the weak basis.
% Möller's algorithm for PIDs returns a strong Gröbner basis.
% Furthermore, Buchberger's coprime and chain criteria can be implemented within this algorithm.

\end{document}